\documentclass[]{interact}

\usepackage{epstopdf}
\usepackage{subfigure}

\usepackage[numbers,sort&compress]{natbib}
\bibpunct[, ]{[}{]}{,}{n}{,}{,}
\makeatletter
\def\NAT@def@citea{\def\@citea{\NAT@separator}}
\makeatother

\theoremstyle{plain}
\newtheorem{theorem}{Theorem}[section]

\theoremstyle{definition}

\theoremstyle{remark}

\def\R{{\mathbb R}}
\def\C{{\mathbb C}}

\def\cl{{C}\!\ell}
\def\U{{\rm U}}
\def\T{{\rm T}}
\def\G{{\rm G}}

\def\diag{{\rm diag}}

\def\SO{{\rm SO}}
\def\GL{{\rm GL}}
\def\SL{{\rm SL}}
\def\O{{\rm O}}
\def\Sp{{\rm Sp}}
\def\sp{\frak{sp}}
\def\u{\frak{u}}
\def\so{\frak{so}}
\def\gl{\frak{gl}}
\def\Spin{{\rm Spin}}
\def\Mat{{\rm Mat}}

\def\tr{{\rm tr}}

\begin{document}

\articletype{}

\title{Classification of Lie algebras of specific type\\ in complexified Clifford algebras}

\author{
\name{D.~S. Shirokov\textsuperscript{a}\textsuperscript{b}\thanks{CONTACT D.~S. Shirokov. Email: dm.shirokov@gmail.com, dshirokov@hse.ru, shirokov@iitp.ru}}
\affil{\textsuperscript{a}National Research University Higher School of Economics, 101000 Moscow, Russia; \textsuperscript{b}Kharkevich Institute for Information Transmission Problems of Russian Academy of Sciences, 127051 Moscow, Russia}
}

\maketitle

\begin{abstract}
We give a full classification of Lie algebras of specific type in complexified Clifford algebras. These sixteen Lie algebras are direct sums of subspaces of quaternion types. We obtain isomorphisms between these Lie algebras and classical matrix Lie algebras in the cases of arbitrary dimension and signature. We present sixteen Lie groups: one Lie group for each Lie algebra associated with this Lie group. We study connection between these groups and spin groups.
\end{abstract}


\section{Introduction}

In this paper, we give a full classification of Lie algebras of specific type in complexified Clifford algebras. These sixteen Lie algebras are direct sums of subspaces of quaternion types suggested by the author in the previous papers \cite{DAN}, \cite{QuatAaca}, \cite{Quat2Aaca}. We obtain isomorphisms between these Lie algebras and classical matrix Lie algebras in the cases of arbitrary dimension and signature.

We present sixteen Lie groups: one Lie group for each Lie algebra associated with this Lie group. In the papers \cite{Lie1} and \cite{Lie2}, we considered five of these sixteen Lie groups and corresponding Lie algebras and obtained isomorphisms with classical matrix Lie groups and Lie algebras. In the current paper, we obtain results for eleven remaining Lie algebras.

In \cite{Lie2}, we studied connection between these groups and spin groups $\Spin_+(p,q)$. In the current paper, we study relation between some of these groups and complex spin groups $\Spin(n, \C)$.

Note that some groups which are considered in the present paper are related to automorphism groups of the scalar products on the spinor spaces (see \cite{Port}, \cite{Lounesto}, \cite{BT}, \cite{Abl3}), but we do not use this fact in the present paper. In \cite{Lounesto}, one found isomorphisms between groups $\G^{12}_{p,q}$, $\G^{23}_{p,q}$, $\G^{12i12}_{p,q}$, $\G^{23i23}_{p,q}$, $\G^{23i01}_{p,q}$ (also in \cite{Snygg}), $\G^{12i03}_{p,q}$ and classical matrix Lie groups. In the present paper, we obtain these isomorphisms and also isomorphisms for the other groups using different techniques based on relations between operations of conjugations in Clifford algebras and corresponding matrix operations. Our main goal is to obtain isomorphisms for corresponding Lie algebras. We use the notion of additional signature of complexified Clifford algebras suggested by the author in the previous paper \cite{Pauli2}.

Let us consider the real Clifford algebra $\cl_{p,q}$ and the complexified Clifford algebra $\C\otimes\cl_{p,q}$, $p+q=n$, $n\geq1$. The constructions of $\cl_{p,q}$ and $\C\otimes\cl_{p,q}$ are discussed in details in \cite{Port} and \cite{Lounesto}.

Let $e$ be an identity element and let $e_a$, $a=1,\ldots,n$ be the generators of $\cl_{p,q}$,
$e_a e_b+e_b e_a=2\eta_{ab}e$, where $\eta=||\eta_{ab}||$ is the diagonal matrix with $+1$ appearing $p$ times on the diagonal and $-1$ appearing $q$ times on the diagonal. The elements
$e_{a_1\ldots a_k}=e_{a_1}\cdots e_{a_k}$, $a_1<\cdots<a_k$, $k=1,\ldots,n$, together with the identity element $e$ form a basis of the Clifford algebra $\cl_{p,q}$.

Let us denote a vector subspace spanned by the elements $e_{a_1\ldots a_k}$ by $\cl^k_{p,q}$.
We have $\cl_{p,q}=\bigoplus_{k=0}^{n}\cl^k_{p,q}.$
Clifford algebra is a $Z_2$-graded algebra and it is represented as the direct sum of even and odd subspaces:
\begin{eqnarray}
\cl_{p,q}=\cl^{(0)}_{p,q}\oplus\cl^{(1)}_{p,q},\quad \cl^{(i)}_{p,q}\cl^{(j)}_{p,q}\subseteq\cl^{(i+j)\rm{mod} 2}_{p,q},\, \cl^{(i)}_{p,q}=\bigoplus_{k\equiv i\,\rm{mod}2}\cl^k_{p,q},\, i, j=0, 1.\label{Z2}
\end{eqnarray}

\section{Lie algebras of specific type in Clifford algebras and Lie groups}

Let us consider $\cl_{p,q}$ as a vector space and represent it in the form of the direct sum of four subspaces of {\it quaternion types} 0, 1, 2 and 3 (see \cite{DAN}, \cite{QuatAaca}, \cite{Quat2Aaca}):
$$\cl_{p,q}=\overline{\textbf{0}}\oplus\overline{\textbf{1}}\oplus\overline{\textbf{2}}\oplus\overline{\textbf{3}},\qquad
\mbox{where}\quad \overline{\textbf{s}}=\bigoplus_{k\equiv s\,\rm{mod} 4}\cl^k_{p,q},\quad s=0, 1, 2, 3.$$
We represent $\C\otimes\cl_{p,q}$ in the form of the direct sum of eight subspaces:
$\C\otimes\cl_{p,q}=\overline{\textbf{0}}\oplus\overline{\textbf{1}}\oplus\overline{\textbf{2}}\oplus\overline{\textbf{3}} \oplus i\overline{\textbf{0}}\oplus i\overline{\textbf{1}}\oplus i\overline{\textbf{2}}\oplus i\overline{\textbf{3}}.$

\begin{theorem}\label{theoremdim} The subspaces $\overline{\textbf{0}}$, $\overline{\textbf{1}}$, $\overline{\textbf{2}}$, and $\overline{\textbf{3}}$ have the following dimensions:
\begin{eqnarray}
\dim \overline{\textbf{0}}=2^{n-2}+2^{\frac{n-2}{2}}\cos{\frac{\pi n}{4}},\qquad
\dim \overline{\textbf{1}}=2^{n-2}+2^{\frac{n-2}{2}}\sin{\frac{\pi n}{4}},\\
\dim \overline{\textbf{2}}=2^{n-2}-2^{\frac{n-2}{2}}\cos{\frac{\pi n}{4}},\qquad
\dim \overline{\textbf{3}}=2^{n-2}-2^{\frac{n-2}{2}}\sin{\frac{\pi n}{4}}.\nonumber
\end{eqnarray}
\end{theorem}
\begin{proof} Using Binomial Theorem
$$(1+i)^n=\sum_{k=0}^n C_n^k i^n=(\sum_{k\equiv 0\!\!\!\!\!\mod 4} C_n^{k}-\sum_{k\equiv 2\!\!\!\!\!\mod 4} C_n^{k})+i(\sum_{k\equiv 1\!\!\!\!\!\mod 4} C_n^{k}-\sum_{k\equiv 3\!\!\!\!\!\mod 4} C_n^{k}),$$
where $C_n^k=\frac{n!}{k! (n-k)!}$ are binomial coefficients, and
$$(1+i)^n=(\sqrt{2}(\cos{\frac{\pi}{4}}+i\sin{\frac{\pi}{4}}))^n=(2)^{\frac{n}{2}}(\cos{\frac{\pi n}{4}}+i\sin{\frac{\pi n}{4}}),$$
we get
\begin{eqnarray}
\sum_{k\equiv 0\!\!\!\!\!\mod 4} C_n^{k}-\sum_{k\equiv 2\!\!\!\!\!\mod 4} C_n^{k}=2^{\frac{n}{2}}\cos{\frac{\pi n}{4}},\qquad
\sum_{k\equiv 0\!\!\!\!\!\mod 1} C_n^{k}-\sum_{k\equiv 3\!\!\!\!\!\mod 4} C_n^{k}=2^{\frac{n}{2}}\sin{\frac{\pi n}{4}}.\nonumber
\end{eqnarray}
Taking into account\footnote{One can easily obtain these expressions using twice Binomial Theorem: $0=(1-1)^n=\sum_{k=0}^n(-1)^k C_n^k$ and $2^n=(1+1)^n=\sum_{k=0}^n C_n^k$.}
\begin{eqnarray}
\sum_{k\equiv 0\!\!\!\!\!\mod 4} C_n^{k}+\sum_{k\equiv 2\!\!\!\!\!\mod 4} C_n^{k}=2^{n-1},\qquad
\sum_{k\equiv 1\!\!\!\!\!\mod 4} C_n^{k}+\sum_{k\equiv 3\!\!\!\!\!\mod 4} C_n^{k}=2^{n-1},\nonumber
\end{eqnarray}
we obtain
$$\dim \overline{\textbf{0}}=\sum_{k\equiv 0\!\!\!\!\!\mod 4} \dim \cl^{k}_{p,q}= \sum_{k\equiv 0\!\!\!\!\!\mod 4} C_n^{k}=2^{n-2}+2^{\frac{n-2}{2}}\cos{\frac{\pi n}{4}}$$
and similarly for the other subspaces.
\end{proof}

Using the method of quaternion typification of Clifford algebra elements, we can find the following Lie algebras. We only want to consider Lie subalgebras that are direct sums of subspaces of quaternion types.

\begin{theorem} The complexified Clifford algebra $\C\otimes\cl_{p,q}$ has the following Lie subalgebras\footnote{Here and below we omit the sign of the direct sum to simplify notation: $\overline{\textbf{0}}\oplus\overline{\textbf{2}}=\overline{\textbf{02}}$, $i\overline{\textbf{1}}\oplus i\overline{\textbf{3}}=i\overline{\textbf{13}}$, $\overline{\textbf{0}}\oplus\overline{\textbf{1}}\oplus\overline{\textbf{2}}\oplus\overline{\textbf{3}}= \overline{\textbf{0123}}$, etc.}
\begin{eqnarray}
&&\overline{\textbf{2}}, \quad \overline{\textbf{02}},\quad \overline{\textbf{12}},\quad \overline{\textbf{23}},\quad \overline{\textbf{2}}\oplus i\overline{\textbf{0}},\quad \overline{\textbf{2}}\oplus i\overline{\textbf{1}},\quad \overline{\textbf{2}}\oplus i\overline{\textbf{2}},\quad \overline{\textbf{2}}\oplus i\overline{\textbf{3}},\quad \overline{\textbf{0123}},
\label{Liealgebras}\\
&&\overline{\textbf{02}}\oplus i\overline{\textbf{02}}, \quad \overline{\textbf{12}}\oplus i\overline{\textbf{12}},\quad \overline{\textbf{23}}\oplus i\overline{\textbf{23}},\quad \overline{\textbf{02}}\oplus i\overline{\textbf{13}},\quad \overline{\textbf{12}}\oplus i\overline{\textbf{03}},\quad \overline{\textbf{23}}\oplus i\overline{\textbf{01}}.\nonumber
\end{eqnarray}
\end{theorem}
Note that the first four subsets $\overline{\textbf{2}}, \overline{\textbf{02}}, \overline{\textbf{12}}, \overline{\textbf{23}}$ are Lie subalgebras of the real Clifford algebra $\overline{\textbf{0123}}=\cl_{p,q}$.
\begin{proof} Using the properties (see \cite{DAN}, \cite{QuatAaca}, \cite{Quat2Aaca})
\begin{eqnarray}
&&[\overline{\textbf{k}},\overline{\textbf{k}}]\subseteq\overline{\textbf{2}},\qquad k=0, 1, 2, 3 \nonumber;\\
&&[\overline{\textbf{k}},\overline{\textbf{2}}]\subseteq\overline{\textbf{k}}, \qquad k=0, 1, 2, 3 \label{1}; \\
&&[\overline{\textbf{0}},\overline{\textbf{1}}]\subseteq\overline{\textbf{3}}, \quad  [\overline{\textbf{0}},\overline{\textbf{3}}]\subseteq\overline{\textbf{1}}, \quad [\overline{\textbf{1}},\overline{\textbf{3}}]\subseteq\overline{\textbf{0}} \nonumber,
\end{eqnarray}
where $[U,V]=UV-VU$ is the commutator of arbitrary Clifford algebra elements $U$ and $V$,
we obtain each set in (\ref{Liealgebras}) is closed with respect to the commutator.
\end{proof}

We can represent these Lie subalgebras in the way as in Figure \ref{figure1}. Every arrow means that one Lie algebra is a Lie subalgebra of the other.

\begin{figure}
\centering
\resizebox*{12cm}{!}{\includegraphics{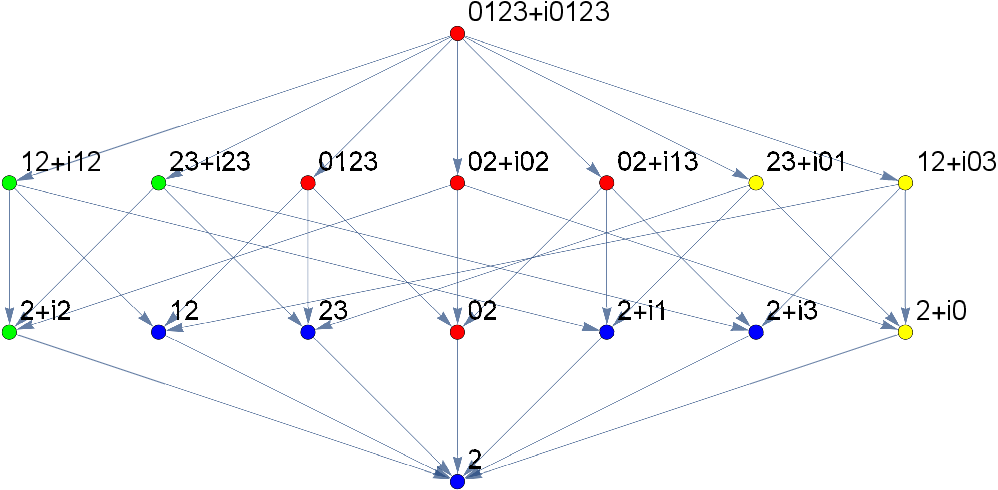}}
\caption{Subspaces of quaternion types as Lie subalgebras of $\C\otimes\cl_{p,q}$.} \label{figure1}
\end{figure}

Any element $U\in\cl_{p,q}$ can be written in the form
\begin{eqnarray}
U=ue+\sum_a u_a e_a+\sum_{a_1<a_2}u_{a_1 a_2}e_{a_1 a_2}+\cdots+u_{1\ldots n}e_{1\ldots n},\label{uform}
\end{eqnarray}
where $u$, $u_a$, $u_{a_1 a_2}$, \ldots, $u_{1\ldots n}$ are real numbers. For arbitrary element $U\in\C\otimes\cl_{p,q}$ we use the same notation (\ref{uform}), where $u, u_a, u_{a_1 a_2}, \ldots, u_{1\ldots n}$ are complex numbers.

Consider the following well-known involutions in $\cl_{p,q}$ and $\C\otimes\cl_{p,q}$:
$$
\hat{U}=U|_{e_a\to-e_a},\quad
\tilde{U}=U|_{e_{a_1\ldots a_r}\to e_{a_r}\ldots e_{a_1}},
$$
where $U$ has the form (\ref{uform}). The operation $U\to \hat{U}$ is called {\it grade involution} and $U\to \tilde{U}$ is called {\it reversion}. Also we have an operation of {\it complex conjugation}
\begin{eqnarray}
\bar U=\bar u e+\sum_a \bar u_a e_a+\sum_{a_1<a_2}\bar u_{a_1 a_2}e_{a_1
a_2}+\sum_{a_1<a_2<a_3}\bar u_{a_1 a_2 a_3}e_{a_1 a_2
a_3}+\cdots +\bar{u}_{1\ldots n}e_{1\ldots n},\nonumber
\end{eqnarray}
where we take the complex conjugation of the complex numbers $u_{a_1 \ldots a_k}$. Superposition of reversion and complex conjugation is {\em pseudo-Hermitian conjugation of Clifford algebra elements}\footnote{The pseudo-Hermitian conjugation of Clifford algebra elements is related to the pseudo-unitary matrix groups as Hermitian conjugation is related to the unitary groups, see \cite{Snygg}, \cite{MarchukShirokov}.}
$$
U^\ddagger=\tilde{\bar{U}}.
$$
In the real Clifford algebra $\cl_{p,q}$, we have $U^\ddagger=\tilde{U}$, because $\bar U=U$.

Note that grade involution and reversion uniquely determine subspaces of quaternion types:
$$\overline{\textbf{s}}=\bigoplus_{k=s {\rm mod}4}\cl^k_{p,q}=\{U\in\cl_{p,q} \,|\, \hat{U}=(-1)^s U, \tilde{U}=(-1)^{\frac{s(s-1)}{2}} U\},\qquad s=0, 1, 2, 3.$$

Now we can consider the following sixteen Lie groups in $\C\otimes\cl_{p,q}$ (see the second column of Table \ref{table1}):
\begin{eqnarray}
&&(\C\otimes\cl_{p,q})^\times,\quad \cl^\times_{p,q},\quad \cl^{(0) \times}_{p,q},\quad (\C\otimes\cl^{(0)}_{p,q})^\times,\quad (\cl^{(0)}_{p,q}\oplus i \cl^{(1)}_{p,q})^\times,\quad \G^{23i01}_{p,q},\nonumber\\
&&\G^{12i03}_{p,q},\quad \G^{2i0}_{p,q},\quad \G^{23i23}_{p,q},\quad \G^{12i12}_{p,q},\quad \G^{2i2}_{p,q},\quad \G^{2i1}_{p,q},\quad \G^{2i3}_{p,q},\quad \G^{12}_{p,q},\quad \G^{23}_{p,q},\quad \G^{2}_{p,q}.\nonumber
\end{eqnarray}

\begin{theorem} The following subsets of $\C\otimes\cl_{p,q}$ in the second column of Table \ref{table1} are Lie groups. The following subsets of $\C\otimes\cl_{p,q}$ in the third column of Table \ref{table1} are Lie algebras of the corresponding Lie groups in the second column of Table \ref{table1}. These Lie groups and Lie algebras have the dimensions given in the forth column of Table \ref{table1}.
\end{theorem}

\begin{table}
\tbl{Lie groups and corresponding Lie algebras of specific type in Clifford algebras.}
{\begin{tabular}{cccc}\toprule
& Lie group & Lie algebra & dimension\\ \midrule
1& $(\C\otimes\cl_{p,q})^\times=\{U\in \C\otimes\cl_{p,q} \, | \, \exists U^{-1}\}$ & $\overline{\textbf{0123}}\oplus i\overline{\textbf{0123}}$ &$2^{n+1}$  \\
2& $\cl^\times_{p,q}=\{U\in\cl_{p,q} \, | \, \exists U^{-1}\}$ & $\overline{\textbf{0123}}$ &$2^n$  \\
3& $\cl^{(0) \times}_{p,q}=\{U\in \cl^{(0)}_{p,q}\, | \, \exists U^{-1}\}$ & $\overline{\textbf{02}}$ &$2^{n-1}$  \\
4& $(\C\otimes\cl^{(0)}_{p,q})^\times=\{ U\in\C\otimes\cl_{p,q}\, | \, \exists U^{-1}\}$ & $\overline{\textbf{02}}\oplus i\overline{\textbf{02}} $ &$2^n$  \\
5& $(\cl^{(0)}_{p,q}\oplus i \cl^{(1)}_{p,q})^\times=\{U\in \cl^{(0)}_{p,q}\oplus i \cl^{(1)}_{p,q} \, | \, \exists U^{-1}\}$ & $\overline{\textbf{02}}\oplus i\overline{\textbf{13}}$ &$2^n$  \\
6& $\G^{23i01}_{p,q}=\{U\in \C\otimes\cl_{p,q}\, | \, U^\ddagger U=e\}$ & $\overline{\textbf{23}}\oplus i\overline{\textbf{01}}$ & $2^n$  \\
7& $\G^{12i03}_{p,q}=\{U\in \C\otimes\cl_{p,q}\, | \, \hat{U}^\ddagger U=e\}$ & $\overline{\textbf{12}}\oplus i\overline{\textbf{03}}$ & $2^n$  \\
8& $\G^{2i0}_{p,q}=\{U\in \cl^{(0)}_{p,q}\, | \, U^\ddagger U=e\}$ & $\overline{\textbf{2}}\oplus i\overline{\textbf{0}}$ & $2^{n-1}$  \\
9& $\G^{23i23}_{p,q}=\{U\in\C\otimes\cl_{p,q}\, | \, \tilde{U}U=e\}$ & $\overline{\textbf{23}}\oplus i\overline{\textbf{23}}$ & $2^{n}-2^{\frac{n+1}{2}}\sin{\frac{\pi (n+1)}{4}}$  \\
10& $\G^{12i12}_{p,q}=\{U\in\C\otimes\cl_{p,q}\, | \, \hat{\tilde{U}}U=e\}$ & $\overline{\textbf{12}}\oplus i\overline{\textbf{12}}$ & $2^{n}-2^{\frac{n+1}{2}}\cos{\frac{\pi (n+1)}{4}}$  \\
11& $\G^{2i2}_{p,q}=\{U\in\C\otimes\cl^{(0)}_{p,q}\, | \, \tilde{U}U=e\}$ & $\overline{\textbf{2}}\oplus i\overline{\textbf{2}}$ & $2^{n-1}-2^{\frac{n}{2}}\cos{\frac{\pi n}{4}}$  \\
12& $\G^{2i1}_{p,q}=\{U\in\cl^{(0)}_{p,q}\oplus i\cl^{(1)}_{p,q} \, | \,  U^\ddagger U=e\}$ & $\overline{\textbf{2}}\oplus i\overline{\textbf{1}}$ & $2^{n-1}-2^{\frac{n-1}{2}}\cos{\frac{\pi (n+1)}{4}}$  \\
13& $\G^{2i3}_{p,q}=\{U\in\cl^{(0)}_{p,q}\oplus i\cl^{(1)}_{p,q} \, | \,  \hat{U}^\ddagger U=e\}$ & $\overline{\textbf{2}}\oplus i\overline{\textbf{3}}$ & $2^{n-1}-2^{\frac{n-1}{2}}\sin{\frac{\pi (n+1)}{4}}$   \\
14& $\G^{23}_{p,q}=\{U\in\cl_{p,q}\, | \, \tilde{U}U=e\}$ & $\overline{\textbf{23}}$ & $2^{n-1}-2^{\frac{n-1}{2}}\sin{\frac{\pi (n+1)}{4}}$  \\
15& $\G^{12}_{p,q}=\{U\in\cl_{p,q}\, | \, \hat{\tilde{U}}U=e\}$ & $\overline{\textbf{12}}$ & $2^{n-1}-2^{\frac{n-1}{2}}\cos{\frac{\pi (n+1)}{4}}$  \\
16& $\G^{2}_{p,q}=\{U\in\cl^{(0)}_{p,q}, | \, \tilde{U}U=e\}$ & $\overline{\textbf{2}}$ & $2^{n-2}-2^{\frac{n-2}{2}}\cos{\frac{\pi n}{4}}$\\  \bottomrule
\end{tabular}}
\label{table1}
\end{table}

\begin{proof} The Lie groups in Table \ref{table1} are subsets of the group $(\C\otimes\cl_{p,q})^\times$ and they are closed under products and inverses. Therefore, they are subgroups of the group $(\C\otimes\cl_{p,q})^\times$.

Let us prove, for example, that $\overline{\textbf{23}}\oplus i\overline{\textbf{23}}$ is a Lie algebra of the corresponding Lie group $\G^{23i23}_{p,q}$. Let $U$ be an arbitrary element of $\G^{23i23}_{p,q}$. Then $U=e+\epsilon u$, where $\epsilon^2=0$ and $u$ is an arbitrary element of the corresponding Lie algebra. Then
$$e=\tilde{U}U=(e-\epsilon \tilde{u})(e+\epsilon u)=e+\epsilon (u-\tilde{u}).$$
Therefore, $u=\tilde{u}$, i.e. $u\in \overline{\textbf{23}}\oplus i\overline{\textbf{23}}$.
We can similarly prove the statement for the other Lie groups and the corresponding Lie algebras.
Using Theorem \ref{theoremdim}, we get the dimension result.
\end{proof}

Note that some Lie groups in the second column of Table \ref{table1} are Lie subgroups of other Lie groups. This property is the same as the corresponding Lie algebras (see Figure \ref{figure1}).

Note that all the Lie groups in Table \ref{table1} contain the spin group $\Spin_+(p,q)$. Similarly, all Lie algebras in Table \ref{table1} contain the Lie algebra $\cl^2_{p,q}$ of the spin group $\Spin_+(p,q)$ because $\cl^2_{p,q}\subset \overline{\textbf{2}}$. We discuss relation between $\Spin_+(p,q)$ and the group $\G^2_{p,q}$ in \cite{Lie2}. In the current paper, we discuss relation between the complex spin group $\Spin(n,\C)$ and the group $\G^{2i2}_{p,q}$ (see below). Note that Salingaros group $\G_{p,q}=\{\pm e, \pm e_{a_1}, \pm e_{a_1 a_2}, \ldots, \pm e_{1\ldots n}\}$ \cite{Sal}, \cite{Abl2}, \cite{Abl3} is a subgroup of $\Spin_+(p,q)$ and all groups in the second column of Table \ref{table1}.

Our goal is to obtain isomorphisms between the sixteen Lie algebras in the third column of Table \ref{table1} and classical matrix Lie algebras. To do this we obtain isomorphisms between Lie groups in the second column of Table \ref{table1} and the corresponding matrix Lie groups.

We have already obtained isomorphisms for the Lie algebras
$\overline{\textbf{2}}\oplus i\overline{\textbf{1}}$, $\overline{\textbf{2}}\oplus i\overline{\textbf{3}}$, $\overline{\textbf{12}}$, $\overline{\textbf{23}}$, $\overline{\textbf{2}}$ with numbers 12-16 in Table \ref{table1} (they are blue in Figure \ref{figure1}). These Lie algebras are isomorphic to the linear, orthogonal, symplectic, and unitary classical Lie algebras in different cases (see papers \cite{Lie1} and \cite{Lie2}). Now we are interested in Lie algebras with numbers 1-11 in Table \ref{table1}.

\section{Lie algebras $\overline{\textbf{0123}}\oplus i\overline{\textbf{0123}}$, $\overline{\textbf{0123}}$, $\overline{\textbf{02}}$, $\overline{\textbf{02}}\oplus i\overline{\textbf{02}}$, $\overline{\textbf{02}}\oplus i\overline{\textbf{13}}$}

Let us consider the Lie algebras $\overline{\textbf{0123}}\oplus i\overline{\textbf{0123}}$, $\overline{\textbf{0123}}$, $\overline{\textbf{02}}$, $\overline{\textbf{02}}\oplus i\overline{\textbf{02}}$, $\overline{\textbf{02}}\oplus i\overline{\textbf{13}}$ with numbers 1-5 in Table \ref{table1} (they are red in Figure \ref{figure1}).

\begin{theorem} We have the following Lie algebra isomorphisms
\begin{eqnarray}
\overline{\textbf{0123}}\oplus i \overline{\textbf{0123}}&\cong&\left\lbrace
\begin{array}{ll}
\gl(2^{\frac{n}{2}},\C), & \mbox{if $n$ is even;}\\
\gl(2^{\frac{n-1}{2}},\C)\oplus \gl(2^{\frac{n-1}{2}},\C), & \mbox{if $n$ is odd,}
\end{array}
\right.\\
\overline{\textbf{0123}}&\cong&\left\lbrace
\begin{array}{ll}
\gl(2^{\frac{n}{2}},\R), & \mbox{if $p-q\equiv0; 2\!\!\mod 8$;}\\
\gl(2^{\frac{n-1}{2}},\R)\oplus \gl(2^{\frac{n-1}{2}},\R), & \mbox{if $p-q\equiv1\!\!\mod 8$;}\\
\gl(2^{\frac{n-1}{2}},\C), & \mbox{if $p-q\equiv3; 7\!\!\mod 8$;}\\
\gl(2^{\frac{n-2}{2}},\mathbb H), & \mbox{if $p-q\equiv4; 6\!\!\mod 8$;}\\
\gl(2^{\frac{n-3}{2}},\mathbb H)\oplus \gl(2^{\frac{n-3}{2}},\mathbb H), & \mbox{if $p-q\equiv5\!\!\mod 8$,}
\end{array}
\right.\\
\overline{\textbf{02}}&\cong&\left\lbrace
\begin{array}{ll}
\gl(2^{\frac{n-1}{2}},\R), & \mbox{if $p-q\equiv1; 7\!\!\mod 8$;}\\
\gl(2^{\frac{n-2}{2}},\R)\oplus \gl(2^{\frac{n-2}{2}},\R), & \mbox{if $p-q\equiv0\!\!\mod 8$;}\\
\gl(2^{\frac{n-2}{2}},\C), & \mbox{if $p-q\equiv2; 6\!\!\mod 8$;}\\
\gl(2^{\frac{n-3}{2}},\mathbb H), & \mbox{if $p-q\equiv3; 5\!\!\mod 8$;}\\
\gl(2^{\frac{n-4}{2}},\mathbb H)\oplus \gl(2^{\frac{n-4}{2}},\mathbb H), & \mbox{if $p-q\equiv4\!\!\mod 8$,}
\end{array}
\right.\\
\overline{\textbf{02}}\oplus i \overline{\textbf{02}}&\cong&\left\lbrace
\begin{array}{ll}
\gl(2^{\frac{n-1}{2}},\C), & \mbox{if $n$ is odd;}\\
\gl(2^{\frac{n-2}{2}},\C)\oplus \gl(2^{\frac{n-2}{2}},\C), & \mbox{if $n$ is even,}
\end{array}
\right.\\
\overline{\textbf{02}}\oplus i \overline{\textbf{13}}&\cong&\left\lbrace
\begin{array}{ll}
\gl(2^{\frac{n}{2}},\R), & \mbox{if $p-q\equiv0; 6\!\!\mod 8$;}\\
\gl(2^{\frac{n-1}{2}},\R)\oplus \gl(2^{\frac{n-1}{2}},\R), & \mbox{if $p-q\equiv7\!\!\mod 8$;}\\
\gl(2^{\frac{n-1}{2}},\C), & \mbox{if $p-q\equiv1; 5\!\!\mod 8$;}\\
\gl(2^{\frac{n-2}{2}},\mathbb H), & \mbox{if $p-q\equiv2; 4\!\!\mod 8$;}\\
\gl(2^{\frac{n-3}{2}},\mathbb H)\oplus \gl(2^{\frac{n-3}{2}},\mathbb H), & \mbox{if $p-q\equiv3\!\!\mod 8$.}
\end{array}
\right.
\end{eqnarray}
\end{theorem}

\begin{proof} We use the following well-known isomorphisms of algebras (\cite{Lounesto}, p.217):
\begin{eqnarray}
\cl_{p,q}&\cong&\left\lbrace
\begin{array}{ll}
\Mat(2^{\frac{n}{2}},\R), & \mbox{ if $p-q\equiv0; 2\!\!\mod 8$;}\\
\Mat(2^{\frac{n-1}{2}},\R)\oplus \Mat(2^{\frac{n-1}{2}},\R), & \mbox{ if $p-q\equiv1\!\!\mod 8$;}\\
\Mat(2^{\frac{n-1}{2}},\C), & \mbox{ if $p-q\equiv3; 7\!\!\mod 8$;}\\
\Mat(2^{\frac{n-2}{2}},\mathbb H), & \mbox{ if $p-q\equiv4; 6\!\!\mod 8$;}\\
\Mat(2^{\frac{n-3}{2}},\mathbb H)\oplus \Mat(2^{\frac{n-3}{2}},\mathbb H), & \mbox{ if $p-q\equiv5\!\!\mod 8$,}
\end{array}
\right.\nonumber\\
\C\otimes\cl_{p,q}&\cong&\left\lbrace
\begin{array}{ll}
\Mat(2^{\frac{n}{2}},\C), & \mbox{if $n$ is even;}\\
\Mat(2^{\frac{n-1}{2}},\C)\oplus \Mat(2^{\frac{n-1}{2}},\C), & \mbox{if $n$ is odd,}
\end{array}
\right.\label{isomcompl}
\end{eqnarray}
and the following isomorphisms
\begin{eqnarray}
\cl_{p,q-1}\cong\cl^{(0)}_{p,q},\qquad \cl^{(0)}_{p,q}\oplus i\cl^{(1)}_{p,q}\cong \cl_{q,p}.\label{isom12}
\end{eqnarray}
To prove the first isomorphism from (\ref{isom12}) we must change the basis of $\cl_{p, q-1}$:
$$ e_a \to e_a e_{n},\qquad a=1, 2, \ldots, n-1,\qquad (e_{n})^2=-e.$$
The elements $e_a e_{n}$, $a=1, 2, \ldots, n-1$ generate $\cl^{(0)}_{p,q}$.

To prove the second isomorphism from (\ref{isom12}) we must change the basis of $\cl_{q,p}$:
$$e_a \to ie_a,\qquad a=1, 2,\ldots, n.$$
Since $(i e_a)^2=-(e_a)^2$, it follows that the signature $(q,p)$ changes to $(p,q)$. Using (\ref{Z2}), we conclude that $\cl_{q,p}$ changes to $\cl^{(0)}_{p,q}\oplus i\cl^{(1)}_{p,q}$ which is closed under multiplication.

Therefore, we obtain the following Lie group isomorphisms
\begin{eqnarray}
(\C\otimes\cl_{p,q})^\times&\cong&\left\lbrace
\begin{array}{ll}
\GL(2^{\frac{n}{2}},\C), & \mbox{if $n$ is even;}\\
\GL(2^{\frac{n-1}{2}},\C)\times \GL(2^{\frac{n-1}{2}},\C), & \mbox{if $n$ is odd,}
\end{array}
\right.\nonumber\\
\cl^\times_{p,q}&\cong&\left\lbrace
\begin{array}{ll}
\GL(2^{\frac{n}{2}},\R), & \mbox{if $p-q\equiv0; 2\!\!\mod 8$;}\\
\GL(2^{\frac{n-1}{2}},\R)\times \GL(2^{\frac{n-1}{2}},\R), & \mbox{if $p-q\equiv1\!\!\mod 8$;}\\
\GL(2^{\frac{n-1}{2}},\C), & \mbox{if $p-q\equiv3; 7\!\!\mod 8$;}\\
\GL(2^{\frac{n-2}{2}},\mathbb H), & \mbox{if $p-q\equiv4; 6\!\!\mod 8$;}\\
\GL(2^{\frac{n-3}{2}},\mathbb H)\times \GL(2^{\frac{n-3}{2}},\mathbb H), & \mbox{if $p-q\equiv5\!\!\mod 8$,}
\end{array}
\right.\nonumber\\
\cl^{(0)\times}_{p,q} &\cong&\left\lbrace
\begin{array}{ll}
\GL(2^{\frac{n-1}{2}},\R), & \mbox{if $p-q\equiv1; 7\!\!\mod 8$;}\\
\GL(2^{\frac{n-2}{2}},\R)\times \GL(2^{\frac{n-2}{2}},\R), & \mbox{if $p-q\equiv0\!\!\mod 8$;}\\
\GL(2^{\frac{n-2}{2}},\C), & \mbox{if $p-q\equiv2; 6\!\!\mod 8$;}\\
\GL(2^{\frac{n-3}{2}},\mathbb H), & \mbox{if $p-q\equiv3; 5\!\!\mod 8$;}\\
\GL(2^{\frac{n-4}{2}},\mathbb H)\times \GL(2^{\frac{n-4}{2}},\mathbb H), & \mbox{if $p-q\equiv4\!\!\mod 8$,}
\end{array}
\right.\nonumber\\
(\C\otimes\cl^{(0)}_{p,q})^\times&\cong&\left\lbrace
\begin{array}{ll}
\GL(2^{\frac{n-1}{2}},\C), & \mbox{if $n$ is odd;}\\
\GL(2^{\frac{n-2}{2}},\C)\times \GL(2^{\frac{n-2}{2}},\C), & \mbox{if $n$ is even,}
\end{array}
\right.\nonumber\\
(\cl^{(0)}_{p,q}\oplus i\cl^{(1)}_{p,q})^\times&\cong&\left\lbrace\begin{array}{ll}
\GL(2^{\frac{n}{2}},\R), & \mbox{if $p-q\equiv0; 6\!\!\mod 8$;}\\
\GL(2^{\frac{n-1}{2}},\R)\times \GL(2^{\frac{n-1}{2}},\R), & \mbox{if $p-q\equiv7\!\!\mod 8$;}\\
\GL(2^{\frac{n-1}{2}},\C), & \mbox{if $p-q\equiv1; 5\!\!\mod 8$;}\\
\GL(2^{\frac{n-2}{2}},\mathbb H), & \mbox{if $p-q\equiv2; 4\!\!\mod 8$;}\\
\GL(2^{\frac{n-3}{2}},\mathbb H)\times \GL(2^{\frac{n-3}{2}},\mathbb H), & \mbox{if $p-q\equiv3\!\!\mod 8$.}
\end{array}
\right.\nonumber
\end{eqnarray}
Using these Lie group isomorphisms, we obtain the Lie algebra isomorphisms of the theorem.
\end{proof}

\section{Theorem on faithful and irreducible representations of complexified Clifford algebras with additional properties}\label{section4}

We need the following theorem to obtain isomorphisms for the groups with numbers 6-11 in Table \ref{table1}. Note that we use a similar method in \cite{Lie1} and \cite{Lie2}. In these papers, we use faithful and irreducible matrix representations of the real Clifford algebras $\cl_{p,q}$ to obtain theorems for the groups with numbers 12-16 in Table \ref{table1}. We use faithful and irreducible matrix representations of the complexified Clifford algebras $\C\otimes\cl_{p,q}$ with some additional properties.

Let us consider a diagonal matrix $J=\diag(1, \ldots, 1, -1, \ldots, -1)$ of arbitrary even size with the same number of $1$'s and $-1$'s on the diagonal. We denote the block-diagonal matrix with two identical blocks $J$ by $\diag(J, J)$.

\begin{theorem} \label{theoremMatrPr} There exists a faithful and irreducible representation of $\C\otimes\cl_{p,q}$ over $\C$ or $\C\oplus\C$
$$\gamma:\C\otimes\cl_{p,q}\to
\left\lbrace
\begin{array}{ll}
\Mat(2^{\frac{n}{2}},\C), & \mbox{if $n$ is even;}\\
\Mat(2^{\frac{n-1}{2}},\C)\oplus \Mat(2^{\frac{n-1}{2}},\C), & \mbox{if $n$ is odd}
\end{array}
\right.
$$
such that
\begin{eqnarray}
(\gamma_a)^\dagger=\eta_{aa}\gamma_a,\qquad a=1, \ldots, n,\label{uslov}
\end{eqnarray}
where $\gamma_a:=\gamma(e_a)$ and ${}^\dagger$ is the Hermitian transpose of a matrix, and
\begin{itemize}
  \item in the case of even $n$, $p\neq0$
  \begin{eqnarray}
  \gamma_{1\ldots p}&=&\alpha_p J,\qquad \alpha_p=\left\lbrace
\begin{array}{ll}
1, & \mbox{if $p\equiv0, 1\!\!\mod 4$};\\
i, & \mbox{if $p\equiv2, 3\!\!\mod 4$},
\end{array}\label{alpha}
\right.
  \end{eqnarray}
  \item in the case of even $n$, $q\neq0$
  \begin{eqnarray}
  \gamma_{p+1\ldots n}&=&\sigma_q J,\qquad \sigma_q=\left\lbrace
\begin{array}{ll}
1, & \mbox{if $q\equiv0, 3\!\!\mod 4$};\\
i, & \mbox{if $q\equiv1, 2\!\!\mod 4$},
\end{array}
\right.\label{sigma}
  \end{eqnarray}
  \item in the case of odd $n\geq 3$, $p\neq0$ is even, and $q$ is odd
  \begin{eqnarray}
  \gamma_{1\ldots p}&=&\alpha_p\diag(J,J),\nonumber
  \end{eqnarray}
  \item in the case of odd $n\geq 3$, $q\neq0$ is even, and $p$ is odd
  \begin{eqnarray}
  \gamma_{p+1\ldots n}&=&\sigma_q\diag(J,J).\nonumber
  \end{eqnarray}
  \end{itemize}
  Moreover, in the last two cases all block-diagonal matrices $\gamma_a$, $a=1, \ldots, n$ consist of two blocks of the same size that differ only in sign.
\end{theorem}
\begin{proof} Let us construct the following representation $\beta$ of $\C\otimes\cl_{p,q}$ over $\C$ or $\C\oplus\C$
$$\beta:\C\otimes\cl_{p,q}\to
\left\lbrace
\begin{array}{ll}
\Mat(2^{\frac{n}{2}},\C), & \mbox{if $n$ is even;}\\
\Mat(2^{\frac{n-1}{2}},\C)\oplus \Mat(2^{\frac{n-1}{2}},\C), & \mbox{if $n$ is odd,}
\end{array}
\right.
$$
using the following algorithm.

For the identity element of $\C\otimes\cl_{p,q}$ we always use the identity matrix $\beta(e)={\bf 1}$ of the corresponding size. For basis element $e_{a_1 \ldots a_k}$ we use the matrix that equals the product of matrices corresponding to $e_{a_1}$, \ldots, $_{a_k}$.

We present the matrix representation $\beta: e_a \to \beta_a$ of $\C\otimes\cl_{n,0}$ below. To obtain the matrix representation of $\C\otimes\cl_{p,q}$, $q\neq 0$, we should multiply matrices $\beta_a$, $a=p+1, \ldots, n$ by $i$.

In the cases of small dimensions, we construct $\beta$ in the following way:
\begin{itemize}
  \item In the case $\C\otimes\cl_{1,0}$: $\beta(e_1)=\diag(1, -1)$.
  \item In the case $\C\otimes\cl_{2,0}$: $\beta(e_1)=\diag(1, -1)$, $\beta(e_2)=\left( \begin{array}{ll}
 0 & 1 \\
 1 & 0 \end{array}\right).$
\end{itemize}
The representation $\beta$ (over $\C$ or $\C\otimes\C$) is faithful and irreducible in these particular cases. Suppose that we have the faithful and irreducible matrix representation $\beta$ of $\C\otimes\cl_{p,q}$ for even $n=p+q=2k$: $\beta(e_a)=\beta_a$, $a=1, \ldots, n$. Then, for the complexified Clifford algebra with $p+q=n+1=2k+1$ we use the following representation:
$e_a\to \diag(\beta_a, -\beta_a)$, $a=1, \ldots, n$, $e_{n+1} \to \diag(i^k \beta_{1} \cdots \beta_n, -i^k \beta_1 \cdots \beta_n)$. For the complexified Clifford algebra with $p+q=n+2=2k+2$ we use the following representation: the same for $e_a$, $a=1, \ldots, n+1$ as in the previous case, and
$e_{n+2}\to\left( \begin{array}{ll}
 0 & {\bf1} \\
 {\bf1} & 0 \end{array}\right).
$

Using this recursive method we obtain the faithful and irreducible representation $\beta$ of all $\C\otimes\cl_{p,q}$ (see isomorphisms (\ref{isomcompl})).

Let us give some examples.

\begin{description}
  \item[$\C\otimes\cl_{3,0}$:]
$e_1 \to \beta_1=\diag(1, -1, -1, 1),\\
  e_2 \to \beta_2=\left( \begin{array}{llll}
 0 & 1 & 0 & 0\\
 1 & 0 & 0 & 0\\
 0 & 0 & 0 & -1\\
 0 & 0 & -1 & 0\end{array}\right),\quad e_3 \to \beta_3=\left( \begin{array}{llll}
 0 & i & 0 & 0\\
 -i & 0 & 0 & 0\\
 0 & 0 & 0 & -i\\
 0 & 0 & i & 0\end{array}\right).$
 \item[$\C\otimes\cl_{4,0}$:]
  $e_1 \to \beta_1,\quad e_2 \to \beta_2,\quad e_3 \to \beta_3,\quad e_4 \to \beta_4=\left( \begin{array}{llll}
 0 & 0 & 1 & 0\\
 0 & 0 & 0 & 1\\
 1 & 0 & 0 & 0\\
 0 & 1 & 0 & 0\end{array}\right).$
  \item[$\C\otimes\cl_{1,3}$:]
  $e_1 \to \beta_1,\quad e_2 \to i\beta_2,\quad e_3 \to i\beta_3,\quad e_4 \to i\beta_4.$
\end{description}

Note that
\begin{eqnarray}
(\beta_a)^\dagger=\eta_{aa}\beta_a,\quad a=1, \ldots, n.\label{uslov2}
\end{eqnarray}

Let us consider the case of even $n$ and the matrix $M=\frac{1}{\alpha_p}\beta_{1\ldots p}$, where $\alpha_p$ is defined in (\ref{alpha}). We have
$$M^2=\frac{1}{\alpha_p^2}(-1)^{\frac{p(p-1)}{2}}\beta_1\ldots\beta_p\beta_p\ldots\beta_1={\bf 1}.$$
Using (\ref{uslov2}), we get $M^{\dagger}=M^{-1}$. Using $M^2={\bf 1}$ and $\tr M=0$\footnote{Because, trace of this matrix equals (up to multiplication by a constant) the projection of element $e_{1\ldots p}$ onto the subspace $\cl^0_{p,q}$  (see \cite{trace}) that is zero.}, we conclude that the spectrum of $M$ consists of the same numbers of $1$'s and $-1$'s.  Therefore there exists unitary matrix $T^\dagger=T^{-1}$ such that $T^{-1}MT=J.$ Now we consider transformation $T^{-1}\beta_a T=\gamma_a$ and obtain another matrix representation $\gamma$ of $\C\otimes\cl_{p,q}$ with $\gamma_{1\ldots p}=T^{-1}\beta_{1\ldots p}T =\alpha_p J$ and (\ref{uslov}) because of (\ref{uslov2}) and $T^{\dagger}=T^{-1}$:
$$(\gamma_a)^\dagger=(T^{-1}\beta_a T)^\dagger=T^\dagger (\beta_a)^\dagger (T^{-1})^\dagger=T^{-1}\eta_{aa}\beta_a T=\eta_{aa}\gamma_a.$$

We can prove the second statement of the theorem similarly. We take $M=\frac{1}{\sigma_q}\beta_{p+1\ldots n}$ and obtain
$$M^2=\frac{1}{\sigma_q^2}(-1)^{\frac{q(q-1)}{2}}(-1)^{q}{\bf 1}={\bf 1}.$$

Let us consider the case of $\C\otimes\cl_{p,q}$ with odd $n=p+q$. Let $p$ be even. We use the faithful and irreducible representation $\beta$ of $\C\otimes\cl_{p,q}$. We have (\ref{uslov2}) and the matrices $\beta_a$ consist of two blocks that differ only in sign. Since $p$ is even, it follows that the matrix $\beta_{1\ldots p}=\diag(D,D)$ consists of two identical blocks which we denote by $D$. Let us consider the matrix $M=\frac{1}{\alpha_p}\beta_{1\ldots p}$. We have $M^2={\bf 1}$, $M^\dagger=M$, and $\tr M=0$. Therefore, $D^2={\bf 1}$, $D^\dagger=D$ and $\tr D=0$. There exists unitary matrix $T_1^\dagger=T_1^{-1}$ such that
$$T_1^{-1}DT_1=J\quad \Rightarrow \quad T^{-1}MT=\diag(J, J),\qquad T=\diag(T_1, T_1).$$

We consider the transformation $T^{-1}\beta_a T=\gamma_a$ and obtain another matrix representation $\gamma$. Since $T^\dagger=T^{-1}$, it follows that $\gamma_a^\dagger=\eta_{aa}\gamma_a$ and the matrices $\gamma_{a}$ consist of two blocks that differ only in sign.

We can prove the last statement of the theorem similarly.
\end{proof}

Note that we can consider in $\C\otimes\cl_{p,q}$ (and $\cl_{p,q}$) a linear operation (involution) $\dagger: \C\otimes\cl_{p,q}\to\C\otimes\cl_{p,q}$ such that $(\lambda e_{a_1 \ldots a_k})^\dagger=\bar{\lambda} (e_{a_1 \ldots a_k})^{-1}$, $\lambda\in\C$. We call this operation {\it Hermitian conjugation of Clifford algebra elements}. This operation is well-known and many authors use it, for example, in different questions of field theory in the case of signature $p=1$, $q=3$. For more details, see \cite{MarchukShirokov}. This operation is called the transposition anti-involution in the case of real Clifford algebras in \cite{Abl1}, \cite{Abl2}, \cite{Abl3}.

Note that we have the following relation between operation of Hermitian conjugation of Clifford algebra elements $\dagger$ and other operations in complexified Clifford algebra $\C\otimes\cl_{p,q}$ (see \cite{MarchukShirokov})
\begin{eqnarray}
U^\dagger&=&(e_{1\ldots p})^{-1}U^\ddagger e_{1 \ldots p},\qquad \mbox{if $p$ is odd;}\nonumber\\
U^\dagger&=&(e_{1\ldots p})^{-1}\hat{U}^{\ddagger} e_{1 \ldots p},\qquad \mbox{if $p$ is even;}\label{sogldager2}\\
U^\dagger&=&(e_{p+1\ldots n})^{-1}U^\ddagger e_{p+1 \ldots n},\qquad \mbox{if $q$ is even;}\nonumber\\
U^\dagger&=&(e_{p+1\ldots n})^{-1}\hat{U}^{\ddagger} e_{p+1 \ldots n},\qquad \mbox{if $q$ is odd.}\nonumber
\end{eqnarray}

The Hermitian conjugation of Clifford algebra elements corresponds to the Hermitian conjugation of matrix $\beta(U^\dagger)=(\beta(U))^\dagger$ for the faithful and irreducible matrix representations over $\C$ and $\C\oplus\C$ of complexified Clifford algebra, based on the fixed idempotent and the basis of the corresponding left ideal, see \cite{MarchukShirokov}. Similarly we have for the matrix representation $\beta$ from Theorem \ref{theoremMatrPr} because of properties (\ref{uslov}).

\section{Lie algebras $\overline{\textbf{23}}\oplus i\overline{\textbf{01}}$, $\overline{\textbf{12}}\oplus i\overline{\textbf{03}}$, $\overline{\textbf{2}}\oplus i\overline{\textbf{0}}$}

Let us consider the Lie algebras $\overline{\textbf{23}}\oplus i\overline{\textbf{01}}$, $\overline{\textbf{12}}\oplus i\overline{\textbf{03}}$, $\overline{\textbf{2}}\oplus i\overline{\textbf{0}}$ with numbers 6-8 in Table~\ref{table1} (they are yellow in Figure \ref{figure1}). One of them, $\G^{23i01}_{p,q}$, has been considered in \cite{Snygg} by Professor J.~Snygg. He calls it c-unitary group. We consider this group in different questions of field theory \cite{pseud} and call it pseudo-unitary group. In \cite{Snygg}, you can find isomorphisms for the group $\G^{23i01}_{p,q}$. In the current paper, we present another proof using relations between matrix operations and operations of conjugation in $\C\otimes\cl_{p,q}$. Also we obtain isomorphisms for the groups $\G^{21i03}_{p,q}$ and $\G^{2i0}_{p,q}$. Finally, we present isomorphisms for the corresponding Lie algebras.

\begin{theorem} We have the following Lie algebra isomorphisms
\begin{eqnarray}
\overline{\textbf{23}}\oplus i\overline{\textbf{01}}&\cong&\left\lbrace
\begin{array}{ll}
\u(2^{\frac{n}{2}}), & \mbox{\rm if $p$ is even and $q=0$;}\\
\u(2^{\frac{n-1}{2}})\oplus \u(2^{\frac{n-1}{2}}), & \mbox{\rm if $p$ is odd and $q=0$;}\\
\u(2^{\frac{n-2}{2}},2^{\frac{n-2}{2}}), & \mbox{\rm if $n$ is even and $q\neq 0$;}\\
\u(2^{\frac{n-3}{2}},2^{\frac{n-3}{2}})\oplus \u(2^{\frac{n-3}{2}},2^{\frac{n-3}{2}}), & \mbox{\rm if $p$ is odd and $q\neq 0$ is even;}\\
\gl(2^{\frac{n-1}{2}}, \C), & \mbox{\rm if $p$ is even and $q$ is odd,}
\end{array}
\right.\\
\overline{\textbf{12}}\oplus i\overline{\textbf{03}}&\cong&\left\lbrace
\begin{array}{ll}
\u(2^{\frac{n}{2}}), & \mbox{\rm if $p=0$ and $q$ is even;}\\
\u(2^{\frac{n-1}{2}})\oplus \u(2^{\frac{n-1}{2}}), & \mbox{\rm if $p=0$ and $q$ is odd;}\\
\u(2^{\frac{n-2}{2}},2^{\frac{n-2}{2}}), & \mbox{\rm if $n$ is even and $p\neq 0$;}\\
\u(2^{\frac{n-3}{2}},2^{\frac{n-3}{2}})\oplus \u(2^{\frac{n-3}{2}},2^{\frac{n-3}{2}}), & \mbox{\rm if $p\neq 0$ is even and $q$ is odd;}\\
\gl(2^{\frac{n-1}{2}}, \C), & \mbox{\rm if $p$ is odd and $q$ is even,}
\end{array}
\right.\\
\overline{\textbf{2}}\oplus i\overline{\textbf{0}}&\cong&\left\lbrace
\begin{array}{ll}
\u(2^{\frac{n-1}{2}}), & \mbox{\rm if $(n,0)$ or $(0,n)$, where $n$ is odd;}\\
\u(2^{\frac{n-2}{2}})\oplus \u(2^{\frac{n-2}{2}}), & \mbox{\rm if $(n,0)$ or $(0,n)$, when $n$ is even;}\\
\u(2^{\frac{n-3}{2}},2^{\frac{n-3}{2}}), & \mbox{\rm if $n$ is odd, $p\neq 0$, and $q\neq 0$;}\\
\u(2^{\frac{n-4}{2}},2^{\frac{n-4}{2}})\oplus \u(2^{\frac{n-4}{2}},2^{\frac{n-4}{2}}), & \mbox{\rm if $p\neq 0$ and $q\neq 0$ are even;}\\
\gl(2^{\frac{n-2}{2}},\C), & \mbox{\rm if $p$ and $q$ are odd.}
\end{array}
\right.
\end{eqnarray}
\end{theorem}

\begin{proof} Let us prove the following Lie group isomorphisms
\begin{eqnarray}
\G^{23i01}_{p,q}\cong\left\lbrace
\begin{array}{ll}
\U(2^{\frac{n}{2}}), & \mbox{\rm if $p$ is even and $q=0$;}\\
\U(2^{\frac{n-1}{2}})\times \U(2^{\frac{n-1}{2}}), & \mbox{\rm if $p$ is odd and $q=0$;}\\
\U(2^{\frac{n-2}{2}},2^{\frac{n-2}{2}}), & \mbox{\rm if $n$ is even and $q\neq 0$;}\\
\U(2^{\frac{n-3}{2}},2^{\frac{n-3}{2}})\times \U(2^{\frac{n-3}{2}},2^{\frac{n-3}{2}}), & \mbox{\rm if $p$ is odd and $q\neq 0$ is even;}\\
\GL(2^{\frac{n-1}{2}}, \C), & \mbox{\rm if $p$ is even and $q$ is odd.}
\end{array}
\right.\nonumber
\end{eqnarray}

In the first two cases ($q=0$), using definition of the group $\G^{23i01}_{p,q}$ (see Table~\ref{table1}) and formulas (\ref{sogldager2}), from $U^\ddagger U=e$  we obtain $U^\dagger U=e$ and an isomorphism with unitary group.

Now we consider the cases $q\geq 1$.
Let $n$ be even. If $p$ and $q$ are odd, then
\begin{eqnarray}
U^\dagger= e_{1\ldots p}U^\ddagger
(e_{1\ldots p})^{-1}\quad \Rightarrow \quad U^\dagger e_{1\ldots p}
U=e_{1\ldots p} U^\ddagger U=e_{1\ldots p}.\nonumber
\end{eqnarray}

We use the first statement of Theorem \ref{theoremMatrPr}. Since $(\gamma_a)^\dagger=\eta_{aa}\gamma_a$, it follows that $\gamma(U^\dagger)=\gamma^\dagger(U)$. We obtain $V^\dagger J V = J$, where $V\in\Mat(2^{\frac{n}{2}},\C)$ and an isomorphism with $\U(2^{\frac{n-2}{2}},2^{\frac{n-2}{2}})$.

In the case of even $p$ and $q$, we have
\begin{eqnarray}
U^\dagger=
e_{p+1\ldots n}U^\ddagger (e_{p+1\ldots n})^{-1}\quad \Rightarrow \quad
U^\dagger e_{p+1\ldots n} U=e_{p+1\ldots n} U^\ddagger
U=e_{p+1\ldots n}.\label{pq2}
\end{eqnarray}
Then we use the second statement of Theorem \ref{theoremMatrPr}.

In the case of odd $p$ and even $q\neq 0$, we have again (\ref{pq2}). We use the forth statement of Theorem \ref{theoremMatrPr}. Every Clifford algebra element has a matrix representation $\diag(R,S)$ with blocks $R$ and $S$ of the same size. We have
$$(\diag(R,S))^\dagger \diag(J, J) \diag(R,S)=\diag(J, J) \Rightarrow R^\dagger J R=J,\, S^\dagger J S=J$$
and obtain an isomorphism with direct sum of two pseudo-unitary groups.

Let us consider the case of even $p$ and odd $q$. If $p\neq 0$, then $$U^\dagger=
e_{1\ldots p}\hat{U}^{\ddagger} (e_{1\ldots p})^{-1}\quad \Rightarrow \quad
\hat{U}^{\dagger} e_{1\ldots p} U=e_{1\ldots p} U^{\ddagger}
U=e_{1\ldots p}.$$
We use the matrix representation $\gamma$ from the third statement of Theorem \ref{theoremMatrPr}. Moreover, we use the fact that $\gamma_a$ are block-diagonal matrices with two blocks that differ in sign. Let the even part of arbitrary element has matrix representation $\diag(A,A)$ and its odd part has matrix representation $\diag(B,-B)$. We obtain
$$(\diag(A-B,A+B))^\dagger \diag(J, J) \diag(A+B,A-B)=\diag(J, J).$$
Equivalently, $(A-B)^\dagger J(A+B)=J$. Therefore, we have $S^\dagger JR=J$ for $R=A+B$ and $S=A-B$. For every matrix $S\in\GL(2^{\frac{n-1}{2}},\C)$ there exists matrix $R=J(S^\dagger)^{-1}J$. We obtain an isomorphism with linear group.

In the cases of signatures $(0,n)$, where $n$ is odd, we use $U^\dagger=\hat{U}^{\ddagger}$ and obtain $\hat{U}^{\dagger}U=e$. Therefore
$$(\diag(A-B,A+B))^\dagger \diag(A+B,A-B)={\bf 1}$$ and $(A-B)^\dagger (A+B)={\bf 1}$. We obtain $S^\dagger R={\bf 1}$ and an isomorphism with linear group again.

We have $\G^{12i03}_{p,q}\cong \G^{23i01}_{q,p}$. To obtain this isomorphism we must change the basis $e_a \to i e_a$, $a=1, \ldots, n$. Note that after this transformation of basis the operation $\tilde{}$ does not change, but the operation $\bar{}$ changes to $\bar{\hat{}}$. Therefore, the operation ${}^\ddagger=\tilde{\bar{}}$ changes to ${\hat{}}^\ddagger=\tilde{\bar{\hat{}}}$ (see definitions of the groups $\G^{12i03}_{p,q}$ and $\G^{23i01}_{p,q}$).

We obtain the following Lie group isomorphisms
\begin{eqnarray}
\G^{12i03}_{p,q}\cong\left\lbrace
\begin{array}{ll}
\U(2^{\frac{n}{2}}), & \mbox{\rm if $p=0$ and $q$ is even;}\\
\U(2^{\frac{n-1}{2}})\times \U(2^{\frac{n-1}{2}}), & \mbox{\rm if $p=0$ and $q$ is odd;}\\
\U(2^{\frac{n-2}{2}},2^{\frac{n-2}{2}}), & \mbox{\rm if $n$ is even and $p\neq 0$;}\\
\U(2^{\frac{n-3}{2}},2^{\frac{n-3}{2}})\times \U(2^{\frac{n-3}{2}},2^{\frac{n-3}{2}}), & \mbox{\rm if $p\neq 0$ is even and $q$ is odd;}\\
\GL(2^{\frac{n-1}{2}}, \C), & \mbox{\rm if $p$ is odd and $q$ is even.}
\end{array}
\right.\nonumber
\end{eqnarray}

We have $\G^{2i0}_{p,q}\cong \G^{12i03}_{p,q-1}\cong \G^{12i03}_{q,p-1}$. To obtain these isomorphisms we must change the basis $e_a \to e_a e_n$, $a=1, \ldots, n-1$, $(e_n)^2=-e$.

We obtain the following Lie group isomorphisms
\begin{eqnarray}
\G^{2i0}_{p,q}\cong\left\lbrace
\begin{array}{ll}
\U(2^{\frac{n-1}{2}}), & \mbox{\rm if $(n,0)$ or $(0,n)$, where $n$ is odd;}\\
\U(2^{\frac{n-2}{2}})\times \U(2^{\frac{n-2}{2}}), & \mbox{\rm if $(n,0)$ or $(0,n)$, when $n$ is even;}\\
\U(2^{\frac{n-3}{2}},2^{\frac{n-3}{2}}), & \mbox{\rm if $n$ is odd, $p\neq 0$, and $q\neq 0$;}\\
\U(2^{\frac{n-4}{2}},2^{\frac{n-4}{2}})\times \U(2^{\frac{n-4}{2}},2^{\frac{n-4}{2}}), & \mbox{\rm if $p\neq 0$ and $q\neq 0$ are even;}\\
\GL(2^{\frac{n-2}{2}},\C), & \mbox{\rm if $p$ and $q$ are odd.}
\end{array}\nonumber
\right.
\end{eqnarray}
Note that $\G^{2i0}_{p,q}\cong\G^{2i0}_{q,p}.$

Using isomorphisms of Lie groups we obtain isomorphisms of the corresponding Lie algebras.
\end{proof}

\section{Lie algebras $\overline{\textbf{23}}\oplus i\overline{\textbf{23}}$, $\overline{\textbf{12}}\oplus i\overline{\textbf{12}}$, $\overline{\textbf{2}}\oplus i\overline{\textbf{2}}$}

Let us consider the Lie algebras $\overline{\textbf{23}}\oplus i\overline{\textbf{23}}$, $\overline{\textbf{12}}\oplus i\overline{\textbf{12}}$, $\overline{\textbf{2}}\oplus i\overline{\textbf{2}}$ with numbers 9-11 in Table \ref{table1} (they are green in Figure \ref{figure1}).

\begin{theorem} \label{theoremOrtCliff} We have the following Lie algebra isomorphisms
\begin{eqnarray}
\overline{\textbf{23}}\oplus i\overline{\textbf{23}}&\cong&\left\lbrace
\begin{array}{ll}
\so(2^{\frac{n}{2}},\C), & \mbox{\rm if $n=0, 2\mod 8$;}\\
\sp(2^{\frac{n-2}{2}},\C), & \mbox{\rm if $n=4, 6\mod 8$;}\\
\so(2^{\frac{n-1}{2}},\C)\oplus \so(2^{\frac{n-1}{2}},\C), & \mbox{\rm if $n=1\mod 8$;}\\
\sp(2^{\frac{n-3}{2}},\C)\oplus \sp(2^{\frac{n-3}{2}},\C), & \mbox{\rm if $n=5\mod 8$;}\\
\gl(2^{\frac{n-1}{2}}, \C), & \mbox{\rm if $n=3, 7\mod 8$,}
\end{array}
\right.\\
\overline{\textbf{12}}\oplus i\overline{\textbf{12}}&\cong&\left\lbrace
\begin{array}{ll}
\so(2^{\frac{n}{2}},\C), & \mbox{\rm if $n=0, 6\mod 8$;}\\
\sp(2^{\frac{n-2}{2}},\C), & \mbox{\rm if $n=2, 4\mod 8$;}\\
\so(2^{\frac{n-1}{2}},\C)\oplus \so(2^{\frac{n-1}{2}},\C), & \mbox{\rm if $n=7\mod 8$;}\\
\sp(2^{\frac{n-3}{2}},\C)\oplus \sp(2^{\frac{n-3}{2}},\C), & \mbox{\rm if $n=3\mod 8$;}\\
\gl(2^{\frac{n-1}{2}}, \C), & \mbox{\rm if $n=1, 5\mod 8$,}
\end{array}
\right.\\
\overline{\textbf{2}}\oplus i\overline{\textbf{2}}&\cong&\left\lbrace
\begin{array}{ll}
\so(2^{\frac{n-1}{2}},\C), & \mbox{\rm if $n=1, 7\mod 8$;}\\
\sp(2^{\frac{n-3}{2}},\C), & \mbox{\rm if $n=3, 5\mod 8$;}\\
\so(2^{\frac{n-1}{2}},\C)\oplus \so(2^{\frac{n-1}{2}},\C), & \mbox{\rm if $n=0\mod 8$;}\\
\sp(2^{\frac{n-3}{2}},\C)\oplus \sp(2^{\frac{n-3}{2}},\C), & \mbox{\rm if $n=4\mod 8$;}\\
\gl(2^{\frac{n-2}{2}}, \C), & \mbox{\rm if $n=2, 6\mod 8$.}
\end{array}
\right.
\end{eqnarray}
\end{theorem}

\begin{proof} Let us prove the following Lie group isomorphisms
\begin{eqnarray}
\G^{23i23}_{p,q}\cong\left\lbrace
\begin{array}{ll}
\O(2^{\frac{n}{2}},\C), & \mbox{\rm if $n=0, 2\mod 8$;}\\
\Sp(2^{\frac{n-2}{2}},\C), & \mbox{\rm if $n=4, 6\mod 8$;}\\
\O(2^{\frac{n-1}{2}},\C)\times \O(2^{\frac{n-1}{2}},\C), & \mbox{\rm if $n=1\mod 8$;}\\
\Sp(2^{\frac{n-3}{2}},\C)\times \Sp(2^{\frac{n-3}{2}},\C), & \mbox{\rm if $n=5\mod 8$;}\\
\GL(2^{\frac{n-1}{2}}, \C), & \mbox{\rm if $n=3, 7\mod 8$.}
\end{array}
\right.\label{is4}
\end{eqnarray}

To prove these Lie group isomorphisms we need the notion of additional signature of $\C\otimes\cl_{p,q}$ suggested by the author in \cite{Pauli2}.

Suppose we have the faithful and irreducible matrix representation $\beta$ over $\C$ or $\C\oplus\C$ of complexified Clifford algebra. We can always use such matrix representation in which all matrices $\beta_a=\beta(e_a)$ are symmetric or skew-symmetric. Let $k$ be the number of symmetric matrices among $\{ \beta_a, a=1, \ldots, n\} $ for the matrix representation $\beta$, and $l$ be the number of skew-symmetric matrices among $\{ \beta_a, a=1, \ldots, n \}$. Let $e_{b_1},\ldots, e_{b_k}$ denote the generators for which the matrices are symmetric. Similarly, we have $e_{c_1},\ldots, e_{c_l}$ for the skew-symmetric matrices.

We use the notion of additional signature of Clifford algebra when we study the relation between matrix representation and operations of conjugation. In complexified Clifford algebra, we have (see \cite{Pauli2})
\begin{eqnarray}
U^\T&=&(e_{b_1 \ldots b_k})^{-1} \tilde{U} e_{b_1\ldots b_k},\qquad \mbox{$k$ is odd;}\nonumber\\
U^\T&=&(e_{b_1 \ldots b_k})^{-1} \tilde{\hat{U}} e_{b_1\ldots b_k},\qquad \mbox{$k$ is even;}\label{sogltransp}\\
U^\T&=&(e_{c_1 \ldots c_l})^{-1} \tilde{U}e_{c_1\ldots c_l},\qquad \mbox{$l$  is even;}\nonumber\\
U^\T&=&(e_{c_1 \ldots c_l})^{-1} \tilde{\hat{U}} e_{c_1\ldots c_l},\qquad \mbox{$l$ is odd,}\nonumber
\end{eqnarray}
where $U^\T=\beta^{-1}((\beta(U))^T)$, and $(\beta(U))^T$ is the transpose of matrix $\beta(U)$.

Numbers $k$ and $l$ depend on the matrix representation $\beta$. But they can take only certain values despite dependence on the matrix representation.

In \cite{Pauli2}, we proved that in a complexified Clifford algebra we have only the following possible values of additional signature as in Table \ref{table2}.

\begin{table}
\tbl{Possible values of additional signature of $\C\otimes\cl_{p,q}$.}
{\begin{tabular}{cc}\toprule
$n\mod 8$& $(k\mod 4, l\mod 4)$\\ \midrule
$0$ & $(0,0),\,(1,3)$  \\
$1$ & $(1,0)$   \\
$2$ & $(1,1),\,(2,0)$  \\
$3$ & $(2,1)$   \\
$4$ & $(3,1),\,(2,2)$   \\
$5$ & $(3,2)$   \\
$6$ & $(3,3),\,(0,2)$   \\
$7$ & $(0,3)$   \\    \bottomrule
\end{tabular}}
\label{table2}
\end{table}

We use the following notation from \cite{Pauli2}. Denote by $[kq]$ the number of symmetric matrices in $\{\beta_a, a=p+1, \ldots, n\}$. Note that this number equals the number of all symmetric and purely imaginary matrices at the same time in $\{\beta_a, a=1,\ldots, n\}$. Denote by $[lp]$ the number of skew-symmetric matrices in $\{\beta_a, a=1,\ldots, p\}$. Note that this number equals the number of all skew-symmetric and purely imaginary at the same time in $\{\beta_a, a=1,\ldots, n\}$. Denote by $[lq]$ the number of skew-symmetric matrices in $\{\beta_a, a=p+1, \ldots, n\}$. This number equals the number of all skew-symmetric and real matrices at the same time in $\{\beta_a, a=1, \ldots, n\}$. Denote by $[kp]$ the number of symmetric matrices in $\{\beta_a, a=1, \ldots, p\}$. Note that this number equals the number of all symmetric and real matrices at the same time in $\{\beta_a, a=1,\ldots, n\}$. We have $n=[kp]+[lp]+[kq]+[lq]$.

Let us return to the proof of Lie group isomorphisms (\ref{is4}). We denote by $\Omega$ the block matrix $$\Omega=\left( \begin{array}{ll}
 0 & -{\bf1} \\
 {\bf1} & 0 \end{array}\right).$$
We use the faithful and irreducible matrix representation $\beta: e_a \to \beta_a$ of $\C\otimes\cl_{p,q}$ from the proof of Theorem \ref{theoremMatrPr} (we can use the matrix representation $\gamma$ from the statement of Theorem \ref{theoremMatrPr} too).

Let us consider the case of even $n$. Let $k$ and $l$ be odd. From (\ref{sogltransp}) we have $U^\T=(e_{b_1 \ldots b_k})^{-1} \tilde{U} e_{b_1\ldots b_k}$. Thus, we obtain for elements of the group $\G^{23i23}_{p,q}$ the condition $U^\T e_{b_1\ldots b_k} U = e_{b_1\ldots b_k}$. Let us consider a real matrix $M=\beta_{b_1 \ldots b_k}$ (in the case of even $[kq]$) or $M=i\beta_{b_1 \ldots b_k}$ (in the case of odd $[kq]$). We have $M^\dagger=M^{-1}$, $\tr M=0$,
$$M^2=(-1)^{[kq]}(\beta_{b_1 \ldots b_k})^2= (-1)^{\frac{k(k-1)}{2}+[kq]+[kq]}{\bf 1}=(-1)^{\frac{k(k-1)}{2}}{\bf 1}.$$
Therefore there exists an orthogonal matrix $T^T=T^{-1}$ such that $T^{-1}M T$ equals $J$ in the case $k=1 \mod 4$ (in the cases $n=0, 2\mod 8$ by Table \ref{table2}) or equals $\Omega$ in the case $k=3 \mod 4$ (in the cases $n=4, 6\mod 8$ by Table \ref{table2}). In both cases, we use transformation $\zeta_a=T^{-1} \beta_a T$ and obtain another matrix representation $\zeta$ such that $\zeta_{b_1 \ldots b_k}$ equals $J$, $iJ$, or $\Omega$, $i\Omega$. Using the fact that $T$ is orthogonal, we conclude that the matrices $\zeta_a$ and $\beta_a$ are both symmetric or antisymmetric for all $a=1, \ldots, n$. Therefore we have the same formulas about the  connection between operations ${}^\T$ and $\tilde{}$. Now we use the matrix representation $\zeta$ and obtain $U^\T J U = J$ in the cases $n=0, 2\mod 8$ or $U^\T \Omega U = \Omega$ in the cases $n=4, 6\mod 8$. We obtain isomorphisms with $\O(2^{\frac{n}{2}},\C)$ (because $\O(a,b,\C)\cong\O(a+b,\C)$) or $\Sp(2^{\frac{n-2}{2}},\C)$.

In the case of even $k$ and even $l\neq 0$, we use $U^\T=(e_{c_1 \ldots c_l})^{-1} \tilde{U} e_{c_1\ldots c_l}$. In the same way, we choose the real matrix $M=\beta_{c_1 \ldots c_l}$ (in the case of even $[lp]$) or $M=i\beta_{c_1 \ldots c_l}$ (in the case of odd $[lp]$). We have $$M^2=(-1)^{[lp]}(-1)^{\frac{l(l-1)}{2}+[lq]}{\bf 1}=(-1)^{\frac{l(l-1)}{2}+l}{\bf 1}=(-1)^{\frac{l(l+1)}{2}}{\bf 1}.$$
Therefore, we obtain an isomorphism with $\O(2^{\frac{n}{2}},\C)$ in the case $l=0 \mod 4$ ($n=0, 2\mod 8$) or an isomorphism with $\Sp(2^{\frac{n-2}{2}},\C)$ in the case $l=3\mod 4$ ($n=4, 6\mod 8$).

In the case of even $k$ and $l=0$ (the cases $n=0, 2\mod 8$), we obtain $U^\T=\tilde{U}$ and an isomorphism with $\O(2^{\frac{n}{2}},\C)$.

Let us consider the case of odd $n$. In the case of odd $k$ and $l=0$ ($n=1\mod 8$), we obtain $U^\T=\tilde{U}$ and an isomorphism with $\O(2^{\frac{n-1}{2}},\C)\times \O(2^{\frac{n-1}{2}},\C)$. Let $k$ be odd and $l\neq 0$ be even (the cases $n=1, 5\mod 8$). We use $U^\T=(e_{c_1 \ldots c_l})^{-1} \tilde{U} e_{c_1\ldots c_l}$. Let us consider the matrix $M=\beta_{c_1 \ldots c_l}$ (in the case of even $[lp]$) or the matrix $M=i\beta_{c_1 \ldots c_l}$ (in the case of odd $[lp]$). This matrix consists of two identical blocks $D$: $M=\diag(D,D)$. We have $M^\dagger=M^{-1}$, $\tr M=0$, and $M^2=(-1)^{[lp]}(-1)^{\frac{l(l-1)}{2}+[lq]}{\bf 1}=(-1)^{\frac{l(l+1)}{2}}{\bf 1}$. We obtain an isomorphism with $\O(2^{\frac{n-1}{2}},\C)\times \O(2^{\frac{n-1}{2}},\C)$ in the case $l=0\mod 4$ ($n=1\mod 8$) or an isomorphism with $\Sp(2^{\frac{n-3}{2}},\C)\times \Sp(2^{\frac{n-3}{2}},\C)$ in the case $l=2\mod 4$ ($n=5\mod 8$).

In the case of even $k$ and odd $l$ ($n=3, 7\mod 8$), we use $U^\T=(e_{b_1 \ldots b_k})^{-1} \hat{\tilde{U}} e_{b_1\ldots b_k}$. Similarly, we obtain $\hat{U}^{\T} J U=J$ or $\hat{U}^{\T} \Omega U=\Omega$. Moreover, we use the fact that $\beta_a$ are block-diagonal matrices with two blocks that differ in sign. The same is true for the matrices $\zeta_a=T^{-1}\beta_a T$ because the matrix $T$ is block-diagonal. Let the even part of arbitrary element has matrix representation $\diag(A,A)$ and its odd part has matrix representation $\diag(B,-B)$. Then
$$\diag(A-B, A+B)^\T \diag(J,J)\diag(A+B, A-B)=\diag(J,J),$$
or, equivalently, $(A-B)^\T J (A+B)=J,$
(or the same equation, where $J$ changes to $\Omega$). In both cases, we obtain an isomorphism with $\GL(2^{\frac{n-1}{2}}, \C)$.

In the case of $k=0$ and odd $l$, we similarly obtain $\hat{U}^{\T} U={\bf 1}$. Then
$$\diag(A-B, A+B)^\T \diag(A+B, A-B)={\bf 1}$$ and
$(A-B)^\T (A+B)={\bf 1}$. We obtain an isomorphism with $\GL(2^{\frac{n-1}{2}}, \C)$.
We can similarly obtain the following Lie group isomorphisms
\begin{eqnarray}
\G^{12i12}_{p,q}\cong\left\lbrace
\begin{array}{ll}
\O(2^{\frac{n}{2}},\C), & \mbox{\rm if $n=0, 6\mod 8$;}\\
\Sp(2^{\frac{n-2}{2}},\C), & \mbox{\rm if $n=2, 4\mod 8$;}\\
\O(2^{\frac{n-1}{2}},\C)\times \O(2^{\frac{n-1}{2}},\C), & \mbox{\rm if $n=7\mod 8$;}\\
\Sp(2^{\frac{n-3}{2}},\C)\times \Sp(2^{\frac{n-3}{2}},\C), & \mbox{\rm if $n=3\mod 8$;}\\
\GL(2^{\frac{n-1}{2}}, \C), & \mbox{\rm if $n=1, 5\mod 8$.}
\end{array}
\right.\nonumber
\end{eqnarray}
We have
$$\G^{12i12}_{p_1, q_1}\cong  \G^{2i2}_{p_2, q_2},\qquad p_1+q_1+1=p_2+q_2.$$
To obtain this isomorphism we must change the basis $e_a \to e_a e_{n}$, $a=1, 2, \ldots, n-1$.

Therefore we have the following Lie group isomorphisms
\begin{eqnarray}
\G^{2i2}_{p,q}\cong\left\lbrace
\begin{array}{ll}
\O(2^{\frac{n-1}{2}},\C), & \mbox{\rm if $n=1, 7\mod 8$;}\\
\Sp(2^{\frac{n-3}{2}},\C), & \mbox{\rm if $n=3, 5\mod 8$;}\\
\O(2^{\frac{n-1}{2}},\C)\times \O(2^{\frac{n-1}{2}},\C), & \mbox{\rm if $n=0\mod 8$;}\\
\Sp(2^{\frac{n-3}{2}},\C)\times \Sp(2^{\frac{n-3}{2}},\C), & \mbox{\rm if $n=4\mod 8$;}\\
\GL(2^{\frac{n-2}{2}}, \C), & \mbox{\rm if $n=2, 6\mod 8$.}
\end{array}
\right.\label{sodspin}
\end{eqnarray}

Using Lie group isomorphisms we obtain isomorphisms of the corresponding Lie algebras.
\end{proof}

Note that the groups $G^{23i23}_{p,q}$, $G^{12i12}_{p,q}$, $G^{2i2}_{p,q}$ and the correponding Lie algebras $\overline{\textbf{23}}\oplus i\overline{\textbf{23}}$, $\overline{\textbf{12}}\oplus i\overline{\textbf{12}}$, $\overline{\textbf{2}}\oplus i\overline{\textbf{2}}$ depend only on $n=p+q$. However, the groups from the previous section $G^{12i03}_{p,q}$, $G^{23i01}_{p,q}$, $G^{2i0}_{p,q}$ and the corresponding Lie algebras $\overline{\textbf{12}}\oplus i\overline{\textbf{03}}$, $\overline{\textbf{23}}\oplus i\overline{\textbf{01}}$, $\overline{\textbf{2}}\oplus i\overline{\textbf{0}}$ depend on $p$ and depend on $q$. They change after the transformation $e_a \to ie_a$.

\section{Relation between $G^{2i2}_{p,q}$ and complex spin groups $\Spin(n,\C)$}

Let us consider the {\it complex spin groups}
$$\Spin(n,\C)=\{U\in(\C\otimes\cl^{(0)}_{p,q})^\times \,|\, \forall x\in\cl^1_{p,q}\quad U^{-1}xU\in\cl^1_{p,q},\, \tilde{U}U=e\}.$$
These groups are subgoups of the groups $\G^{2i2}_{p,q}$, $\G^{12i12}_{p,q}$, $\G^{23i23}_{p,q}$, $(\C\otimes\cl^{(0)}_{p,q})^\times$, and $(\C\otimes\cl_{p,q})^\times$.

The groups $\Spin(n,\C)$ are double covers of $\SO(n,\C)$ (see \cite{Lounesto}). It is well-known that $\Spin(n,\C)$ is isomorphic to the following classical matrix Lie groups (see \cite{Lounesto})
\begin{eqnarray}
\Spin(n,\C)\cong\left\lbrace
\begin{array}{ll}
\{\pm 1\}, & \mbox{\rm if $n=0$;}\\
\O(1,\C), & \mbox{\rm if $n=1$;}\\
\GL(1, \C), & \mbox{\rm if $n=2$;}\\
\Sp(2,\C), & \mbox{\rm if $n=3$;}\\
\Sp(2,\C)\times \Sp(2,\C), & \mbox{\rm if $n=4$;}\\
\Sp(4,\C), & \mbox{\rm if $n=5$;}\\
\SL(4, \C), & \mbox{\rm if $n=6$.}
\end{array}
\right.\nonumber
\end{eqnarray}

Note that $\Spin(n,\C)$ coincides with $\G^{2i2}_{p,q}$ in the cases $n\leq 5$. In the case $n=6$, the condition $\forall x\in\cl^1_{p,q}$,\, $U^{-1}xU\in\cl^1_{p,q}$ from the definition of $\Spin(6,\C)$ leads to the condition $\det\gamma(U)=1$ for the matrix representation $\gamma$ and we obtain $\SL(4,\C)$ (not $\GL(4,\C)$ as for the group $\G^{2i2}_{p,q}$).

Note that in the cases $n\geq 6$ $\Spin(n,\C)$ is a subgroup of $\G^{2i2}_{p,q}$. Therefore we know the classical matrix groups (\ref{sodspin}) that contain $\Spin(n,\C)$ in the cases $n\geq 6$.

\section*{Acknowledgements}

The author thanks the Referees for their constructive remarks and comments.

\section*{Funding}

The article was prepared within the framework of the Academic Fund Program at the National Research University Higher School of Economics (HSE) in 2017-2018 (grant 17-01-0009) and by the Russian Academic Excellence Project ``5-100''.


\end{document}